\documentclass[10pt]{article}

\usepackage{graphicx}%
\usepackage{amsmath,amssymb,amsthm,upref,bm}%
\usepackage{labelfig}%
\usepackage{mathrsfs}

\DeclareMathAlphabet{\varmathbb}{U}{pxsyb}{m}{n}

\newtheorem{proposition}{Proposition}%
\newtheorem{definition}{Definition}%

\newcommand{\MF}[1]{\mathop{#1\vrule height0.45ex width0pt}\nolimits}

\newcommand{\pd}[3][]{\mathchoice{\raise-0.5pt\hbox{$\partial$}%
\vphantom{\partial}_{\mkern-1.5mu#2}^{\mkern0.4mu#1}\mkern0.3mu}%
{\raise-0.5pt\hbox{$\partial$}%
\vphantom{\partial}_{\mkern-1.5mu#2}^{\mkern0.4mu#1}\mkern0.3mu}%
{\raise-0.5pt\hbox{$\scriptstyle\partial$}%
\vphantom{\partial}_{\mkern-1.7mu#2}^{\mkern0.1mu#1}\mkern0.1mu}%
{\raise-0.5pt\hbox{$\scriptscriptstyle\partial$}%
\vphantom{\partial}_{\mkern-1.7mu#2}^{\mkern0.1mu#1}\mkern0.1mu}#3}
 
\newcommand{\D}{\mathrm{d}\kern0.2pt}%
\newcommand{\ii}{\kern0.05em\mathrm{i}\kern0.05em}% 
\newcommand{\E}[1]{\textrm{e}^{#1}}%
\renewcommand{\vec}[1]{\bm{#1}}%
\newcommand{\Cb}{\varmathbb{C}}%

\newcommand\mhat[1]{\widehat{#1}}

\def\transp{\mathsf{T}}

\begin{document}

\baselineskip=4.4mm

\makeatletter

\title{\bf On freely floating bodies trapping \\ time-harmonic waves in water \\
covered by brash ice}

\author{Nikolay Kuznetsov and Oleg Motygin}

\date{}

\maketitle

\vspace{-10mm}

\begin{center}
Laboratory for Mathematical Modelling of Wave Phenomena, \\ Institute for Problems
in Mechanical Engineering, Russian Academy of Sciences, \\ V.O., Bol'shoy pr. 61,
St. Petersburg 199178, Russian Federation \\ E-mail:
nikolay.g.kuznetsov@gmail.com\,, o.v.motygin@gmail.com
\end{center}

\begin{abstract}
A mechanical system consisting of water covered by brash ice and a body freely
floating near equilibrium is considered. The water occupies a half-space into which
an infinitely long surface-piercing cylinder is immersed, thus allowing us to study
two-dimensional modes of the coupled motion which is assumed to be of small
amplitude. The corresponding linear setting for time-harmonic oscillations reduces
to a spectral problem whose parameter is the frequency. A constant that
characterises the brash ice divides the set of frequencies into two subsets and the
results obtained for each of these subsets are essentially different.

For frequencies belonging to a finite interval adjacent to zero, the total energy of
motion is finite and the equipartition of energy holds for the whole system. For
every frequency from this interval, a family of motionless bodies trapping waves is
constructed by virtue of the semi-inverse procedure. For sufficiently large
frequencies outside of this interval, all solutions of finite energy are trivial.
\end{abstract}

\setcounter{equation}{0}

\section{Introduction}

This paper continues the rigorous study (initiated in \cite{NGK10}) of a freely
floating rigid bodies trapping time-harmonic waves in an inviscid, incompressible,
heavy fluid, say water (see also \cite{KM,KM1,KM2} and \cite{K1}). We consider the
infinitely deep water in irrotational motion bounded from above by a free surface
unbounded in all horizontal directions, but unlike the cited papers dealing with the
open surface, we assume here that it is totally covered with the brash ice. The body
is supposed to be an infinitely long cylinder which allows us to consider
two-dimensional modes orthogonal to cylinder's generators. We also assume that the
body is surface-piercing and unaffected by all external forces (for example due to
constraints on its motion) except for gravity. The motion of the whole system is
supposed to be of small amplitude near equilibrium, and so the linear model
developed by John \cite{John1} is used to describe the coupled motion of water and
body. However, the free-surface boundary condition must be amended to take into
account the presence of the brash ice covering the water. Such a condition was
proposed by Peters \cite{P}, who considered the brash ice as an infinitely thin mat
whose particles do not interact; that is, the only forces acting on it are those due
to gravity and the pressure of the water from below (see also \cite{GS}).

Our aim is twofold: first, we apply the so-called semi-inverse procedure for the
construction of motionless two-dimensional bodies each trapping a time-harmonic
mode; that is, the water covered by the brash ice and the body oscillate at the same
frequency, whereas the energy of water and ice motion is finite. We recall that a
trapped mode in two and three dimensions describes a free oscillation of the system
that has finite energy (see \cite{LM}, p.~17). This should be distinguished from edge
waves and waves trapped around an array of cylinders.

It should be mentioned that the same semi-inverse procedure was used by Kuzne\-tsov
\cite{NGK10} for obtaining two-dimensional trapping bodies that are motionless in
the open water (it is outlined in \S 4.1). However, there is an essential
distinction between the case of the open water and the water covered by the brash
ice. Namely, no restriction was imposed on the trapping frequencies in the former
case. On the contrary, trapping frequencies are bounded by a constant characterising
the brash ice and this bound is essential in the latter case. Indeed, our second
result shows that if a solution of the problem has finite energy and the frequency
exceeds a certain bound, then this solution is trivial. Generally, this bound is
larger than the constant mentioned above, but these bounds may coincide for some
bodies. It is essential that this result is valid for all bodies.

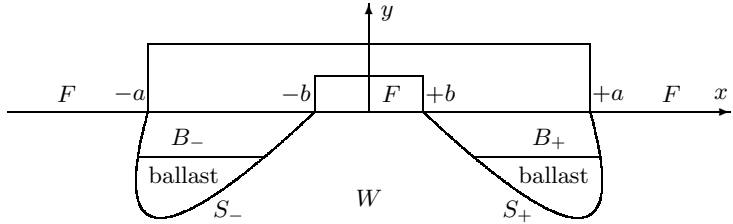
\begin{figure}
\begin{center}
%\vspace{2mm}
\unitlength=0.6mm

\special{em:linewidth 0.4pt}

\linethickness{0.4pt}

\begin{picture}(163.00,62.00)(-3,90)

\put(28.0,130.00){\line(0,1){15.00}}

\put(65.0,130.00){\line(0,1){8.00}}

\put(89.0,130.00){\line(0,1){8.00}}

\put(126.0,130.00){\line(0,1){15.00}}

\put(28.0,145.00){\line(1,0){98.00}}

\put(65.0,138.00){\line(1,0){24.00}}

\put(100.0,120.00){\line(1,0){28.00}}

\put(26.0,120.00){\line(1,0){28.00}}

\put(-3,130.00){\vector(1,0){160.00}}

\put(77.00,130.00){\vector(0,1){24.00}}

\bezier{464}(28.00,130.00)(15.00,83.00)(65.00,130.00)

\multiput(28,130)(-6,-6){5}{\line(-1,-1){5.1}}

\multiput(126,130)(6,-6){5}{\line(1,-1){5.1}}

\bezier{464}(89.00,130.00)(139.00,83.00)(126.00,130.00)

\put(10.00,134.00){\makebox(0,0)[cc]{{\small $F$}}}

\put(144.00,134.00){\makebox(0,0)[cc]{{\small $F$}}}

\put(81.00,152.00){\makebox(0,0)[cc]{{\small $y$}}}

\put(92.8,134.00){\makebox(0,0)[cc]{{\small $+b$}}}

\put(130.00,134.00){\makebox(0,0)[cc]{{\small $+a$}}}

\put(24.00,134.00){\makebox(0,0)[cc]{{\small $-a$}}}

\put(60.8,134.00){\makebox(0,0)[cc]{{\small $-b$}}}

\put(37.00,124.00){\makebox(0,0)[cc]{{\small $B_-$}}}

\put(36.00,116.00){\makebox(0,0)[cc]{{\small ballast}}}

\put(118.00,116.00){\makebox(0,0)[cc]{{\small ballast}}}

\put(117.00,124.00){\makebox(0,0)[cc]{{\small $B_+$}}}

\put(46.00,108.00){\makebox(0,0)[cc]{{\small $S_-$}}}

\put(110.00,108.00){\makebox(0,0)[cc]{{\small $S_+$}}}

\put(155.00,134.00){\makebox(0,0)[cc]{{\small $x$}}}

\put(77.00,111.00){\makebox(0,0)[cc]{{\small $W$}}}

\put(82.00,134.00){\makebox(0,0)[cc]{{\small $F$}}}

\end{picture}

\end{center}
\vspace{-8mm} \caption{A definition sketch of the cylinder cross-section with
immersed parts denoted by $B_-$ and $B_+$; their wetted boundaries are $S_-$ and
$S_+$ respectively. The cross-section of the infinitely thin layer of the brash ice
covering the free surface is denoted by $F$; it consists of three parts two of which
lie on the $x$-axis outside $|x| > a$ and the third one is between $x = -b$ and $x =
+b$; $W$ is the cross-section of the water domain.} \vspace{-2mm}
\end{figure}

\section{Statement of the problem}

Let the Cartesian coordinate system $(x,y)$ in a plane orthogonal to the generators
of a freely floating infinitely long cylinder be chosen so that the $y$-axis is
directed upwards, whereas the mean free surface of the water intersects this plane
along the $x$-axis, and so the cross-section $W$ of the water domain is a subset of
$\mathbb{R}^2_- = \{ x \in \mathbb{R}, \, y < 0 \}$. Let $\mhat{B}$ denote the
bounded two-dimensional domain whose closure is the cross-section of a floating
cylinder in its equilibrium position (see figure~1).

We suppose that $\mhat{B} \setminus \overline{\mathbb{R}^2_-}$\,---\,the part of the
body located above the surface of water covered by the brash ice\,---\,is a nonempty
domain, whereas the immersed part $B = \mhat{B} \cap \mathbb{R}^2_-$ is the union of
a finite number of domains (in particular, it may be a single domain). Thus, $D =
\mhat{B} \cap \partial \mathbb{R}^2_-$ consists of the same number of nonempty
intervals of the $x$-axis; see figure 1, where $D = \{ x \in (-a, -b) \cup (b, a) ,
\, y=0 \}$ in the case of two immersed parts. We suppose that $W$ is $\mathbb{R}^2_-
\setminus \skew2\overline{B}$. Furthermore, we assume that $W$ is a Lipschitz
domain, and so the unit normal $\vec{n}$ pointing to the exterior of $W$ is defined
almost everywhere on $\partial W$. Finally, we denote by $S = \partial \mhat{B} \cap
\mathbb{R}^2_-$ the wetted curve (the number of its components is equal to the
number of immersed domains), whereas $F = \partial \mathbb{R}^2_- \setminus \skew2
\overline {D}$ is the infinitely thin layer of the brash ice covering the free
surface at rest.

To describe the small-amplitude coupled motion of the system it is standard to apply
the linear setting, in which case the first-order unknowns are used. These are the
velocity potential $\MF{\Phi} (x,y;t)$ and the vector column $\MF{\vec{q}} (t)$
describing the motion of the body. This vector has the following three components:

\vspace{1mm}

\noindent $\bullet$ $q_1$ and $q_2$ are the displacements of the centre of mass in
the horizontal and vertical directions respectively from its rest position $\bigl(
x^{(0)}, y^{(0)} \bigr)$;

\noindent $\bullet$ $q_3$ is the angle of rotation about the axis that goes through
the centre of mass orthogonally to the $(x,y)$-plane (the angle is measured from the
$x$- to the $y$-axis).

\vspace{1mm}

On the surface $F$, the brash ice is characterised by a non-negative function
$\sigma$ equal to the ratio of the local area density of the brash ice to the
constant volume density of the water. The values of $\sigma$ are less than or equal
to one, and we assume that $\sigma$ is constant throughout $F$; the value of this
constant is defined by the water salinity.

We omit relations governing the general time-dependent motion (see the details in
\cite{NGK10}) and turn directly to the time-harmonic oscillations of the system
for which purpose we use the ansatz
\begin{equation}
 \bigl( \MF{\Phi} (\vec{x},y,t), \vec{q}(t) \bigr) = {\rm Re} \bigl\{ \E{-\ii\omega t} 
 \bigl( \MF{\varphi} (\vec{x},y), \ii \vec{\chi} \bigr) \bigr\} ,
 \label{eq:ansatz}
\end{equation}
where $\omega > 0$ is the radian frequency, $\varphi \in \MF{H^1_{loc}}(W)$ is a
complex-valued function bounded at infinity and $\vec{\chi} \in \Cb^3$. Then the
governing relations for $\bigl( \varphi, \vec{\chi} \bigr)$ are as follows:
\begin{gather}
 \nabla^2 \varphi = 0 \quad \mbox{in} \ W ;
 \label{eq:1}\\
 \pd{y} \varphi = \nu ( \varphi + \sigma \pd{y} \varphi ) \quad \mbox{on} \ F , 
 \quad \mbox{where} \ \nu = \omega^2 / g ;
 \label{eq:2}\\
 \pd{\vec{n}} \varphi = \omega \, \vec{\chi}^\transp \vec{N} \ \Big(\!\! = \omega 
 \sum_1^3 N_j \chi_j \Big) \quad \mbox{on} \ S ;
 \label{eq:4}\\
 \nabla \varphi \to 0 \quad \mbox{as} \ y \to -\infty ;
 \label{eq:3}\\
 \omega^2 \bm{E} \vec{\chi} = - \omega \int_{S} \varphi \vec{N} \,
 \D{}s + g \, \bm{K} \vec{\chi} . \label{eq:5}
\end{gather}
Here $\nabla = (\pd{x}, \pd{y})$ is the spatial gradient and $g > 0$ is the
acceleration due to gravity that acts in the direction opposite to the $y$-axis;
$\vec{N} = (N_1, N_2, N_3)^\transp$ (the operation $^\transp$ transforms a vector
row into a vector column and vice versa), where $(N_1, N_2)^\transp = \vec{n}$, $N_3
= \left( x - x^{(0)}, y - y^{(0)} \right)^\transp \times \vec{n}$ and $\times$
stands for the vector product. In the equations of the body motion \eqref{eq:5}, the
$3\!\times\!3$ matrices are as follows:
\begin{equation}
 \MF{\bm{E}} =
\begin{pmatrix}
 I^M & 0 & 0 \\
 0 & I^M & 0 \\
 0 & 0 & I^M_2
\end{pmatrix} \quad {\rm and} \quad \MF{\bm{K}} =
\begin{pmatrix}
 0 & 0 & 0 \\
 0 & I^D & I^D_x \\
 0 & I^D_x & I^D_{xx} + I^S_y
\end{pmatrix} .
\label{eq:EK}
\end{equation}
The positive elements of the mass/inertia matrix $\MF{\bm{E}}$ are
\[ I^M = \rho_0^{-1} \int_{\mhat{B}} \MF{\rho}(x,y) \, \D x \D y \ \ \mbox{and} 
\ \ I^M_2 = \rho_0^{-1} \int_{\mhat{B}} \MF{\rho}(x,y) \Big[ \left( x - x^{(0)}
\right)^2 + \left( y - y^{(0)} \right)^2 \Big] \D x \D y ,
\]
where $\MF{\rho}(x,y) \geq 0$ is the density distribution within the body and
$\rho_0 > 0$ is the constant density of water. The first term on the right-hand side
of \eqref{eq:5} is due to the hydrodynamic pressure, whereas the second one is
related to the buoyancy (see, for example, \cite{John1}). The non-zero elements
of the matrix $\bm{K}$ are
\begin{gather*}
I^D = \int_D \D x > 0, \quad I^D_x = \int_D \big( x - x^{(0)} \big) \D x , \\
I^D_{xx} = \int_D \big( x - x^{(0)} \big)^2 \D x > 0 , \quad I^S_y = \int_S \big( y
- y^{(0)} \big) \D x \, \D y .
\end{gather*}
It should be noted that the matrix $\bm{K}$ is symmetric.

First, we suppose that $\nu \sigma \in [0, 1)$, which restricts the range of
frequencies because $\sigma$ is a given constant. Under this assumption, the problem
is studied in \S\S 3 and 4; the case when $\nu \sigma \geq 1$ is considered in \S 5.
For $\nu \sigma \in [0, 1)$ the boundary condition \eqref{eq:2} is equivalent to
\begin{equation}
 \pd{y} \varphi = \nu^{(\sigma)} \varphi \quad \mbox{on} \ F , \quad \mbox{where}
 \ \nu^{(\sigma)} = \nu / ( 1 - \nu \sigma ) > \nu . 
\label{eq:2'}
\end{equation}
This boundary condition is of the same form as in the case of the open water when
the coefficient $\nu$ stands on the right-hand side instead of $\nu^{(\sigma)}$. The
parameter $\nu$ has the following expression in terms of $\nu^{(\sigma)}$:
\[ \nu = \nu^{(\sigma)} / \big( 1 + \sigma \nu^{(\sigma)} \big) < \nu^{(\sigma)} . \]

As in the problem describing waves on the open water, it is natural to complement
the Laplace equation \eqref{eq:1} and the boundary condition \eqref{eq:2'} by the
following radiation condition (it means that the potential \eqref{eq:ansatz}
describes outgoing waves):
\begin{equation}
 \int_{W\cap\{|x|=b\}} \bigl| \pd{|x|} \varphi - \ii \nu^{(\sigma)} \varphi \bigr|^2\,
 \D{}s = \MF{o}(1) \quad \mbox{as} \ b \to \infty .
 \label{eq:6}
\end{equation}
In relations \eqref{eq:4}, \eqref{eq:5}, \eqref{eq:2'} and \eqref{eq:6}, $\omega$ is
a spectral parameter (into \eqref{eq:2'} and \eqref{eq:6} it is involved through
$\nu^{(\sigma)}$), which is sought together with the eigenvector
$(\varphi,\vec{\chi})$.

Since $W$ is a Lipschitz domain and $\varphi \in \MF{H^1_{loc}}(W)$, it is natural
to understand the problem, namely relations \eqref{eq:1}, \eqref{eq:4} and
\eqref{eq:2'}, in the sense of the integral identity
\begin{equation}
 \int_{W} \nabla \varphi \nabla \psi \,\D x \D{}y = \nu^{(\sigma)}
 \int_{F} \varphi \, \psi \, \D x + \omega \int_{S} \psi \, 
 \vec{N}^\transp \vec{\chi} \, \D{}s ,
\label{eq:intid}
\end{equation}
which must hold for an arbitrary smooth $\psi$ having a compact support in
$\overline W$, whereas the remaining conditions \eqref{eq:3} and \eqref{eq:6}
specify the behaviour of $\varphi$ at infinity.

The problem formulated above must be augmented by the following subsidiary
conditions concerning the equilibrium position (see \cite{John1}):

\vspace{1mm}

\noindent $\bullet$ $I^M = \int_B \D x \D{}y$ (Archimedes' law\,---\,the mass of the
displaced liquid is equal to that of the body); 

\noindent $\bullet$ $\int_B \bigl(x - x^{(0)} \bigr) \, \D x \D{}y = 0$ (the centre
of buoyancy lies on the same vertical line as the centre of mass); 

\noindent $\bullet$ the $2 \times 2$ matrix $\bm{K}'$ that stands in the lower right
corner of $\bm{K}$ is positive definite.

\vspace{1mm}

\noindent The last of these requirements yields the stability of the body
equilibrium position, which is understood in the classical sense that any
instantaneous infinitesimal disturbance causes the position changes which remain
infinitesimal, except for purely horizontal drift, for all subsequent times.

\section{Equipartition of energy}

It is known (see, for example, \cite[\S\,2.2.1]{LWW}), that a potential
satisfying relations \eqref{eq:1}, \eqref{eq:3}, \eqref{eq:2'} and \eqref{eq:6} has
the same asymptotic representation at infinity as Green's function, namely
\begin{eqnarray}
&& \ \ \ \ \ \ \ \ \MF{\varphi} (x, y) = \MF{A}_\pm (y) \, \E{\ii \nu^{(\sigma)}
|x|} + \MF{r}_\pm (x, y) , \nonumber \\ && |r_\pm|^2, \, |\nabla r_\pm| = \MF{O}
\bigl(
 [x^2 + y^2]^{-1} \bigr) \ \mbox{as} \ x^2 + y^2 \to \infty . \label{eq:uas}
\end{eqnarray}
Moreover, the following equality holds for the coefficients
\begin{equation}
 \nu^{(\sigma)} \int_{-\infty}^0 \left( |\MF{A}_+ (y)|^2 + |\MF{A}_- (y)|^2 \right)
 \D{} y = - {\rm Im} \int_S \overline{\varphi}\,\pd{\vec{n}} \varphi \, \D{}s .
\label{eq:ener}
\end{equation}

Assuming that $\bigl( \MF{\varphi} , \vec{\chi} \bigr)$ is a solution of the problem
formulated in section~2, we rearrange the last formula by virtue of the coupling
conditions \eqref{eq:4} and \eqref{eq:5}, thus obtaining (see details in \cite{K1},
\S 3):
\begin{equation}
 \nu^{(\sigma)} \int_{-\infty}^0 \left( |\MF{A}_+ (y)|^2 + |\MF{A}_- (y)|^2 \right)
 \D{} y = {\rm Im} \Bigl\{ \omega^2 \, \overline{\vec{\chi}}^\transp \bm{E} \vec{\chi}
 - g \overline{\vec{\chi}}^\transp \bm{K} \vec{\chi} \Bigr\} .
\label{eq:ener'}
\end{equation}
In the same way as in \cite{KM,KM1}, this yields the following assertion about the
kinetic and potential energy of the water motion.

\begin{proposition}\label{propos:1}
Let\/ $\bigl( \MF{\varphi} , \vec{\chi} \bigr)$ be a solution of problem\/
\eqref{eq:1}--\eqref{eq:5}, then
\begin{equation}
 \int_W |\nabla \varphi|^2\,\D{}x \, \D{}y < \infty \quad \mbox{and} \quad
 \nu^{(\sigma)} \int_F |\varphi|^2 \, \D{}x < \infty \, , \label{eq:finenerg}
\end{equation}
that is, $\varphi \in H^1 (W)$. Moreover, the following equality holds:
\begin{equation}
 \int_W |\nabla \varphi|^2\,\D{}x \,\D{}y + \omega^2 \overline{\vec{\chi}}^\transp
 \bm{E} \vec{\chi} = \nu^{(\sigma)} \int_F |\varphi|^2 \, \D{}x + g \,
 \overline{\vec{\chi}}^\transp \bm{K} \vec{\chi} .
\label{eq:lagrange}
\end{equation}
\end{proposition}

Here the kinetic energy of the water/body system stands on the left-hand side,
whereas we have the potential energy of the coupled motion on the right-hand side.
The latter takes into account that the water is covered by the brash ice, and so the
last formula generalises the energy equipartition equality valid when a body is
freely floating in the open water.

Proposition 1 shows that if $( \varphi, \vec{\chi} )$ is a solution with
complex-valued components, then its real and imaginary parts separately satisfy the
problem. This allows us to consider $( \varphi, \vec{\chi} )$ as an element of the
real product space $H^1 (W) \times \mathbb{R}^3$ in what follows; an equivalent norm
in $H^1 (W)$ is defined by the sum of two quantities \eqref{eq:finenerg}.

\begin{definition}\rm
Let the subsidiary conditions concerning the equilibrium position (see \S~2) hold
for the freely floating body $\mhat{B}$. A non-trivial real solution $( \varphi,
\vec{\chi} ) \in H^1 (W) \times \mathbb{R}^3$ of the problem \eqref{eq:intid},
\eqref{eq:5} and \eqref{eq:6} is called a {\it mode trapped}\/ by this body, whereas
the corresponding value of $\omega$ is referred to as a {\it trapping frequency}.
\end{definition}

\section{Motionless bodies trapping waves \\ in the presence of the brash ice}

In \cite{NGK10}, a semi-inverse procedure was used for the construction of
motionless bodies freely floating in the open water of infinite depth and trapping a
two-dimensional mode at the frequency $\omega$ given arbitrarily. The idea of the
procedure is to seek bodies for a prescribed trapped mode. Subsequently, this
approach was developed in \cite{KM1,KM2}, where various axisymmetric trapping
bodies, motionless and heaving, were constructed as well as sets of multiple bodies
some of which are motionless, whereas the others heave. In this section, we first
apply the same idea to construct a family of motionless bodies trapping waves in the
water covered by the brash ice and then consider the effect of the brash ice by
comparing these bodies with those trapping waves in the open water.

\subsection{Construction of a family of motionless trapping bodies}

Let $\omega > 0$ be fixed so that $\nu = \omega^2 / g$ satisfies the inequality $\nu
< \sigma^{-1}$, where the constant $\sigma > 0$ is given, thus restricting the range
of admissible frequencies. In order to obtain a family of immersed parts of
motionless bodies that trap waves in the water covered by the brash ice we follow
the considerations used in \S\S 3 and 4 of \cite{NGK10} with $\nu$ changed to
$\nu^{(\sigma)} = \nu / ( 1 - \nu \sigma )$.

Let us consider the following motionless mode $(\varphi_0, (\bm{0}, 0)^\transp)$,
where $\bm{0} \in \mathbb{R}^2$ denotes the zero-displacement vector, $0$ stands for
the zero angle of rotation and
\[ \varphi_0 (x,y) = g^2 \phi_0 (\nu^{(\sigma)} x, \nu^{(\sigma)} y) / \omega^3 \, , \]
whereas the non-dimensional velocity potential $\phi_0$ is as follows:
\begin{equation}
\phi_0 (\nu^{(\sigma)} x, \nu^{(\sigma)} y) = \int_0^\infty \frac{k \, \E{k
\nu^{(\sigma)} y}}{k-1} \left[ \sin k (\nu^{(\sigma)} x - \pi) - \sin k
(\nu^{(\sigma)} x + \pi) \right] \, \D k \, .
\label{eq:phi0}
\end{equation}
Here, the integral is understood as usual improper integral because its integrand is
bounded; indeed, the location of zeroes coincides for the denominator and numerator.
Moreover, we have that
\[ \phi_0 (\nu^{(\sigma)} x, \nu^{(\sigma)} y) = \left( 2 \, \nu^{(\sigma)} \right)^{-1}
\left[ G_x (x,y; - \pi / \nu^{(\sigma)} , 0) - G_x (x,y; \pi / \nu^{(\sigma)}, 0)
\right] ,
\]
where $G (x,y; \xi,\eta)$ is Green's function of the time-harmonic water-wave
problem (see \cite{LWW}, \S 1.2.1, where the properties of $G$ are described).
Therefore,
\begin{equation}
\partial_y \phi_0 - \nu^{(\sigma)} \phi_0 = 0 \quad \mbox{for} \ y=0 \ \mbox{and} \
\nu^{(\sigma)} x \neq \pm \pi ,
\label{eq:F0}
\end{equation}
and
\begin{eqnarray}
&& \phi_0 (x,y) = O \left( [ x^2 + y^2]^{-1} \right) , \nonumber \\ && |\nabla
\phi_0 (x,y)| = O \left( [x^2 + y^2]^{-3/2} \right)
\label{eq:inf0}
\end{eqnarray}
as $x^2 + y^2 \to \infty$. These estimates yield that the kinetic and potential
energy is finite for $\phi_0$ in every domain away from the singularities of this
function.

The next crucial point of the inverse procedure is to use streamlines corresponding
to the velocity potential \eqref{eq:phi0} in order to define two immersed contours
of a freely floating trapping body that is symmetric about the $y$-axis and has
$x^{(0)} = 0$. The last property is guaranteed by a proper choice of density
distribution which is always possible for a symmetric body. Moreover, it is possible
to choose this distribution so that $y^{(0)}$ is sufficiently close to the ordinate
of the lowest points of the body, thus yielding that all of the subsidiary
conditions hold.

Taking a harmonic conjugate to $\phi_0$ in the form
\begin{equation}
\psi_0 (\nu^{(\sigma)} x, \nu^{(\sigma)} y) = \int_0^\infty \frac{k \, \E{k
\nu^{(\sigma)} y}}{k-1} \left[ \cos k (\nu^{(\sigma)} x - \pi) - \cos k
(\nu^{(\sigma)} x + \pi) \right] \, \D k \, ,
\label{eq:psi0}
\end{equation}
we see that this stream function decays at infinity. Let us list several other
properties of $\psi_0$ (see their proof in \cite{LWW}, pp.~178--179) used for
construction of a family of bodies trapping the mode $(\varphi_0, (\bm{0},
0)^\transp)$. Since $\psi_0$ is an odd function of $x$, we formulate these
properties only for $x > 0$.

\begin{figure}
%\vspace{2mm}
\begin{center}
\SetLabels
 \L (0.555*0.66) $\nu^{(\sigma)} x$ \\
 \L (0.525*0.94) $\nu^{(\sigma)} y$ \\
 \L (0.244*0.46) ballast \\
 \L (0.170*0.46) ${\cal S}_-$ \\
 \L (0.708*0.46) ballast \\
 \L (0.834*0.46) ${\cal S}_+$ \\
\endSetLabels \leavevmode
\strut\AffixLabels{\includegraphics[width=120mm]{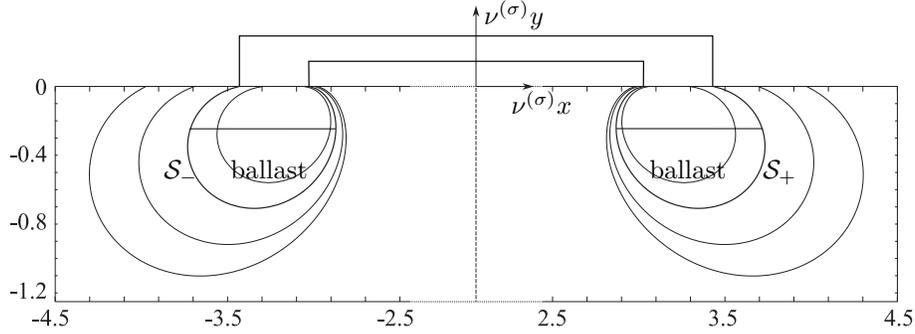}}
\end{center}
\vspace{-2mm} \caption{Level lines of the stream function $\psi_0$ plotted in
non-dimensional coordinates. An example of the wetted contours bounding a motionless
trapping body is shown in bold and denoted by ${\cal S}_-$ and ${\cal S}_+$; other
pairs of symmetric streamlines can also be considered as boundaries of immersed
parts of trapping bodies. In order to contract the width of the figure, a reduced
horizontal scale is applied on the interval (-2.5, 2.5).} \vspace{-2mm}
\label{fig:examples}
\end{figure}

The trace $\psi_0 (\nu^{(\sigma)} x,0)$ has only one positive zero $\nu^{(\sigma)}
x_0 \in \big( \frac{2 \pi}{3} , \pi \big)$ and
\[ \psi_0 (\nu^{(\sigma)} x,0) < 0 \quad \mbox{on} \ (0, \nu^{(\sigma)} x_0) .
\]
Moreover, this function increases monotonically from $0$ to $+\infty$ on $\left[
\nu^{(\sigma)} x_0, \pi \right)$ and decreases monotonically from $+\infty$ to $0$
on $\left( \pi , +\infty \right)$. Therefore, for every non-dimensional $d > 0$ the
streamline ${\cal S}_+ = \big\{ (\nu^{(\sigma)} x, \nu^{(\sigma)} y) \in
\overline{\mathbb{R}^2_-} : \, x > 0, \, y < 0, \, \psi_0 (\nu^{(\sigma)} x,
\nu^{(\sigma)} y) = d \big\}$ (see the right part of figure~2) connects two points
that lie on the positive $\nu^{(\sigma)} x$-axis on either side of the point $(\pi,
0)$ at which $\psi_0$ is infinite.

Thus, for the chosen $\omega$ and every $d > 0$ the domain $B_+$ between
\[ S_+ = \left\{ (x,y) \in \overline{\mathbb{R}^2_-} : (\nu^{(\sigma)} x, 
\nu^{(\sigma)} y) \in {\cal S}_+ \right\}
\]
and the $x$-axis serves as the right immersed part of a single motionless trapping
body $\mhat{B}$ obtained by connecting $B_+$ with the symmetric domain $B_-$ as
shown in figure~2 for the corresponding domains in non-dimensional variables.
Indeed, the Cauchy--Riemann equations imply that $\varphi_0$ satisfies the
homogeneous Neumann condition on $S = S_- \cup S_+$, and so the coupling condition
\eqref{eq:4} is fulfilled for the mode $(\varphi_0 , (\bm{0}, 0)^\transp)$. The
coupling condition \eqref{eq:5} is also fulfilled for this mode because it reduces
to the equality
\[ \int_{S_- \cup S_+} \varphi_0 \, n_y \, \D s = 0 \, .
\]
Indeed, the other two equalities are trivial by the symmetry of $S_- \cup S_+$ and
the fact that the corresponding integrands are odd functions of $x$. The proof that
the last integral vanishes is based on the second Green's formula and the asymptotic
formulae \eqref{eq:inf0} (see details in \cite{NGK10}).

\subsection{The effect of the brash ice on motionless trapping bodies}

The procedure described in \S 4.1 rigorously establishes the existence of trapping
bodies, but it does not present pictorially the effect of the brash ice, that is,
how these bodies differ from those trapping waves in the open water. In order to do
this, we apply rescaling of specially chosen contours from the non-dimensional
coordinates $(\nu^{(\sigma)} x, \nu^{(\sigma)} y)$ to $(\nu x, \nu y)$ for which
purpose the equality $\nu = \nu^{(\sigma)} / ( 1 + \sigma \nu^{(\sigma)})$ serves.
Rescaling can be realised in two different manners\,---\,by keeping the area of the
immersed part fixed and by keeping the length of $D$ fixed. The corresponding
results are shown in figures 3 and 4 respectively, where three values of frequency
are considered to demonstrate the effect; the frequency is characterised by the
product $\nu \sigma$.

In figures 3 and 4 (a)--(c), examples of symmetric bodies trapping waves in the
water covered by the brash ice are shown. The bodies are represented by the
boundaries of their right immersed parts; see the left contour in each figure (cf.
figure~2, where an example of the whole body is shown in the non-dimensional
coordinates $(\nu^{(\sigma)} x, \nu^{(\sigma)} y)$). The location of the body that
traps waves in the open water at the same frequency is indicated for comparison; it
is plotted in bold (see the contour on the right-hand side of each figure). Both
bodies have either the same area of the immersed part (the scaling adopted in figure
3) or the same length of $D$ (the scaling adopted in figure 4). The dashed lines in
figures 3 (b) and 4 (b) show the contours from which the left ones are obtained by
the scaling $(\nu^{(\sigma)} x, \nu^{(\sigma)} y) \mapsto (\nu x, \nu y)$.

Thus, figures 3 (a)--(c) demonstrate the following properties of the obtained
symmetric trapping bodies with the fixed immersed area. The spacing between two
immersed parts is less for bodies that trap waves in the water covered by the brash
ice than for bodies that trap waves in the open water. Moreover, this spacing
decreases as the frequency of trapped waves increases, whereas the length of $D$, on
the opposite, increases together with the frequency.

The properties of symmetric trapping bodies with the fixed length of $D$ are as
follows; see figures 4 (a)--(c). The spacing between two immersed parts demonstrates
the same behaviour as in the other case, whereas the immersed area decreases as the
frequency increases.
 
\begin{figure}
\vspace{2mm}
\begin{center}
\SetLabels
 \L (0.88*-0.011) $\nu x$ \\
 \L (0.0*0.22) $\nu y$ \\
 \L (0.88*0.34) $\nu x$ \\
 \L (0.0*0.57) $\nu y$ \\
 \L (0.88*0.695) $\nu x$ \\
 \L (0.0*0.925) $\nu y$ \\
 \L (-0.04*0.98) $(a)$ \\
 \L (-0.04*0.627) $(b)$ \\
 \L (-0.04*0.278) $(c)$ \\
\endSetLabels \leavevmode
\strut\AffixLabels{\includegraphics[width=120mm]{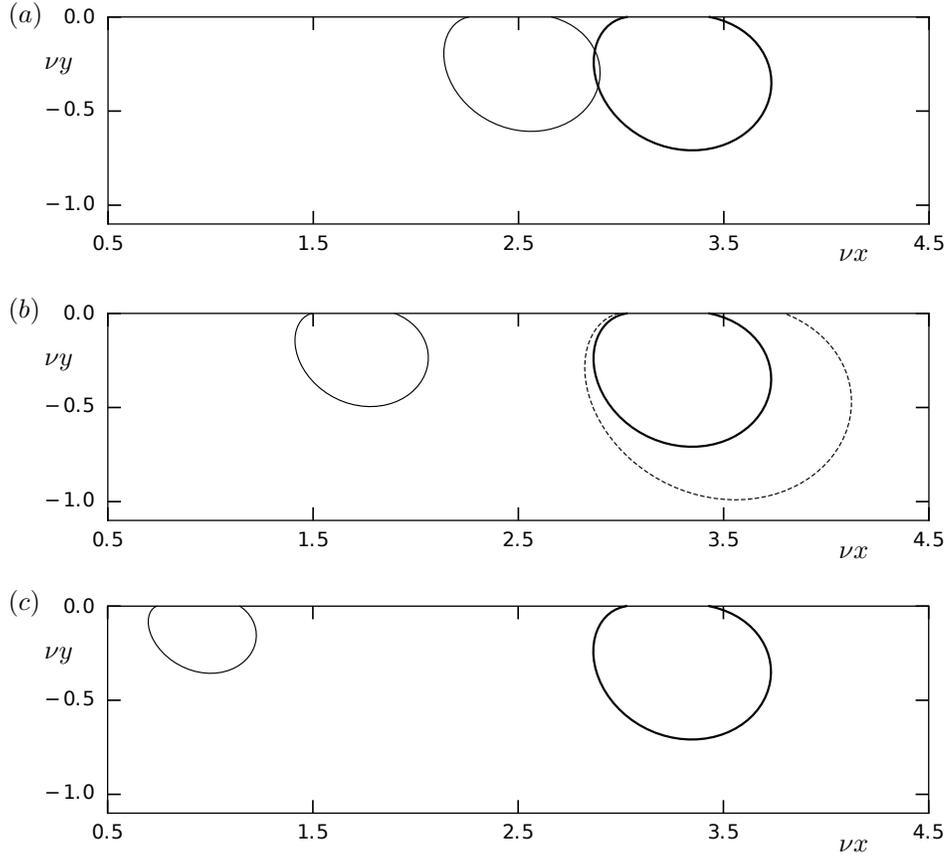}}
\end{center}
\vspace{2mm} \caption{Examples of bodies trapping waves in the water covered by the
brash ice for $\nu \sigma = 1/4$ (a); $\nu \sigma = 1/2$ (b); $\nu \sigma = 3/4$
(c). In each figure, the left contour is the right immersed part of a symmetric body
trapping waves; the bold contour shows the location of the body that has the same
length of $D$ and traps waves in the open water at the same frequency. The dashed
line in figure (b) shows the contour from which the left one is obtained by the
scaling $(\nu^{(\sigma)} x, \nu^{(\sigma)} y) \mapsto (\nu x, \nu y)$.}
\vspace{-2mm}
\label{fig:contours2}
\end{figure}

\begin{figure}
\vspace{2mm}
\begin{center}
\SetLabels
 \L (0.95*0.0) $\nu x$ \\
 \L (-0.01*0.255) $\nu y$ \\
 \L (0.95*0.34) $\nu x$ \\
 \L (-0.01*0.595) $\nu y$ \\
 \L (0.95*0.682) $\nu x$ \\
 \L (-0.01*0.94) $\nu y$ \\
 \L (-0.06*0.98) $(a)$ \\
 \L (-0.06*0.634) $(b)$ \\
 \L (-0.06*0.295) $(c)$ \\
\endSetLabels \leavevmode
\strut\AffixLabels{\includegraphics[width=120mm]{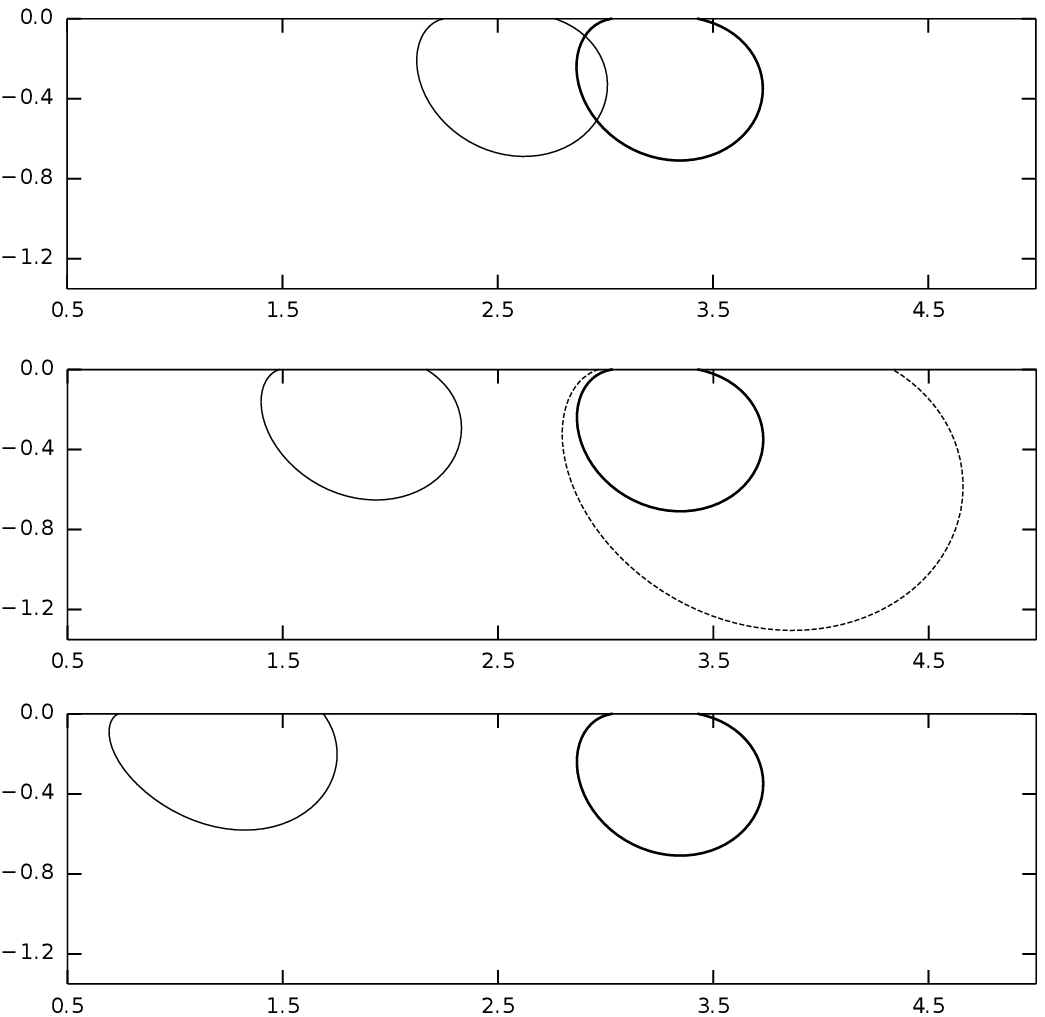}}
\end{center}
\vspace{1mm} \caption{Examples of bodies trapping waves in the water covered by the
brash ice for $\nu \sigma = 1/4$ (a); $\nu \sigma = 1/2$ (b); $\nu \sigma = 3/4$
(c). In each figure, the left contour is the right immersed part of a symmetric body
trapping waves; the bold contour shows the location of the body that has the same
area and traps waves in the open water at the same frequency. The dashed line in
figure (b) shows the contour from which the left one is obtained by the scaling
$(\nu^{(\sigma)} x, \nu^{(\sigma)} y) \mapsto (\nu x, \nu y)$.} \vspace{-3mm}
\label{fig:contours1}
\end{figure}

\section{The effect of the brash ice when $\bm{\nu \sigma \geq 1}$}

Let $\mhat{B}$ be a freely floating body, that is, it satisfies all subsidiary
conditions of \S 2. For the matrices $\bm{E}$ and $\bm{K}$ corresponding to
$\mhat{B}$ (see \eqref{eq:EK} and the following formulae that define the matrices'
elements), we denote by $\lambda_*$ the largest $\lambda$ such that $\det( \lambda
\bm{E} - \bm{K}) = 0$.

\begin{proposition} 
Let $( \varphi, \vec{\chi} ) \in H^1 (W) \times \mathbb{R}^3$ be a solution of
problem \eqref{eq:1}--\eqref{eq:5}. If $\omega$ is such that $\nu \geq \max \{
\lambda_*, \sigma^{-1} \}$, where $\lambda_*$ corresponds to $\mhat{B}$ through
$\bm{E}$ and $\bm{K}$, then $( \varphi, \vec{\chi} )$ is a trivial solution.
\end{proposition}

\begin{proof}
Let $R_{\alpha, \beta} = \{ (x, y) \in \mathbb{R}^2_- : |x| < \alpha , \ y > -\beta
\}$, where $\alpha, \beta > 0$ are sufficiently large; for example, such that $B
\subset R_{\alpha-1, \beta-1}$. Applying the first Green's formula, we write
\begin{eqnarray}
\int_{W \cap R_{\alpha, \beta}} \!\!\!\!\!\!\!\! | \nabla \varphi |^2 \, \D x \D y =
\int_S \varphi \, \pd{\vec{n}} \varphi \, \D s + \int_{-\alpha}^\alpha [ \varphi (x,
y) \, \pd{y} \varphi  (x, y)]_{y=-\beta}^{y=0, x \in F} \, \D x \nonumber \\ +
\int_{-\beta}^0 [ \varphi (x, y) \, \pd{x} \varphi (x, y) ]_{x=-\alpha}^{x=\alpha}
\, \D y , \label{eq:fgf}
\end{eqnarray}
where the Laplace equation is taken into account. The assumptions imposed on
$\varphi$ yield that: 1) there exists a sequence $\{ \alpha_k \}_{k=1}^{\infty}$
such that $\alpha_k \to \infty$ as $k \to \infty$ and the last integral with $\beta
= \infty$ and $\alpha = \alpha_k$ tends to zero as $k \to \infty$; 2) for all
$\alpha > 0$ we have that $\int_{-\alpha}^\alpha [ \varphi (x, -\beta) \, \pd{y}
\varphi (x, -\beta)] \, \D x \to 0$ as $\beta \to \infty$. Passing to the limit in
\eqref{eq:fgf}, first as $\beta \to \infty$ and then as $\alpha_k \to \infty$, we
obtain
\begin{equation}
\int_{W} | \nabla \varphi |^2 \, \D x \D y - \int_F \varphi (x, 0) \, \pd{y} \varphi
(x, 0) \, \D x = \int_S \varphi \, \pd{\vec{n}} \varphi \, \D s .
\label{eq:lim}
\end{equation}
Here the integral over $F$ vanishes when $\nu = \sigma^{-1}$ (in this case, the
boundary condition \eqref{eq:2} reduces to $\varphi = 0$ on $F$); if $\nu >
\sigma^{-1}$, this integral is equal to $\nu^{(\sigma)} \int_F [\varphi (x, 0)]^2
\, \D x$, which is a consequence of \eqref{eq:2'}. In both cases, the left-hand side
of \eqref{eq:lim} is positive for a non-trivial $\varphi$ (in the second case,
because $\nu^{(\sigma)} < 0$).

On the other hand, it follows from \eqref{eq:4} and \eqref{eq:5} that
\[ \int_S \varphi \, \pd{\vec{n}} \varphi \, \D s = \omega \vec{\chi}^\transp 
\int_S \varphi \, \vec{N} \, \D s = - g^{-1} \vec{\chi}^\transp [\nu \bm{E} -
\bm{K}] \vec{\chi} \, ,
\]
and so the right-hand side term of \eqref{eq:lim} is non-positive when $\nu \geq
\lambda_*$. Thus, this contradicts to the positivity of the left-hand side with a
non-trivial $\varphi$ when $\nu \geq \max \{ \lambda_*, \sigma^{-1} \}$. This
completes the proof of the proposition.
\end{proof}

\section{Concluding remarks} 

This note extends our previous work on motionless trapping bodies. Here, it has
been shown that there exist two-dimensional bodies with a vertical axis of symmetry
and two immersed parts and the following properties. These bodies are freely
floating but motionless and trap some time-harmonic modes provided their frequencies
are bounded by a constant characterising the brash ice. Thus, there is a coupled
motion of the water covered by this ice that does not radiate waves to infinity in
the presence of such a body. Therefore, in the absence of viscosity, this
oscillation will persist forever.

On the other hand, for every freely floating body there exists a bound for
frequencies (it depends also on the constant that characterises the brash ice) such
that no modes are trapped at the frequencies exceeding this bound. Unlike uniqueness
theorems known hitherto and concerning time-harmonic waves in the presence of freely
floating bodies, this result does not involve any geometric restrictions.

\end{document}